\documentclass[journal]{IEEEtran}
\usepackage{amsmath,amsfonts,amsthm, amssymb, latexsym} 
\usepackage{dsfont}
\usepackage[english]{babel} 
\usepackage{fancyhdr} 
\usepackage{cite}
\usepackage{macros}
\usepackage{bbold}
\usepackage{graphicx}
\usepackage{epstopdf}
\usepackage{epsfig}
\DeclareGraphicsExtensions{.pdf,.jpeg,.png,.jpg,.eps}
\usepackage{subcaption}

\newtheorem{theorem}{Theorem}
\newtheorem{lemma}{Lemma}

\newtheorem{proposition}{Proposition}

\title{Wireless Networks of Bounded Capacity}
\author{Grace Villacr\'es and Tobias Koch}

\date{}

\begin{document}

\onecolumn

\maketitle

\begin{abstract}
We consider a noncoherent wireless network, where the transmitters and receivers are cognizant of the statistics of the fading coefficients, but are ignorant of their realizations. We demonstrate that if the nodes do not cooperate, if they transmit symbols that all follow the same distribution, and if the variances of the fading coefficients decay exponentially or more slowly, then the channel capacity is bounded in the SNR. This confirms a similar result by Lozano, Heath, and Andrews.

\renewcommand{\thefootnote}{}
\footnote{The authors are with the Signal Theory and Communications Department, Universidad Carlos III de Madrid, 28911, Legan\'es, Spain and with the Gregorio Mara\~n\'on Health Research Institute (e-mail: \texttt{\{gvillacres,koch\}@tsc.uc3m.es}).\\
This work has been partly supported by the Ministerio de Economía of Spain (projects ’COMONSENS’, id. CSD2008-00010 and ’OTOSiS’, id. TEC2013-41718-R), the Comunidad de Madrid (project ’CASI-CAM- CM’, id. S2013/ICE-2845) and the European Community's Seventh Framework Programme (FP7/2007-2013) under Grant Nr. 333680.}
\end{abstract}
\setcounter{footnote}{0}

\section{Introduction}
\IEEEPARstart{T}{he} information-theoretical limits of wireless networks have been studied extensively in the past and it has been demonstrated that cooperation between nodes in wireless networks can substantially increase the channel capacity. Indeed, theoretical findings suggest that throughput gains of up to two orders of magnitude are obtainable by exploiting techniques such as Coordinated Multi-Point (CoMP), network Multiple Input Multiple Output (MIMO) or Interference Alignment (IA), provided that the communicating nodes have access to the fading coefficients in the network \cite{simeone_local_2009,ozgur_hierarchical_2007,cadambe_interference_2008}. Knowledge of these fading coefficients is usually referred to as \emph{channel-state information (CSI)}, and the assumption that the nodes have perfect knowledge of the fading coefficients is usually referred to as \emph{perfect CSI}. However, it is \emph{prima facie} unclear whether perfect CSI can be obtained in practical systems. In fact, over-the-air trials have only demonstrated disappointingly low throughput improvements (not exceeding 30\% in some scenarios) \cite{irmer_coordinated_2011}.

The purpose of this work is to investigate the channel capacity of wireless networks when the nodes neither have perfect CSI nor do they perform a channel estimation to obtain information on the fading coefficients. This is relevant, e.g., if the nodes do not have the computational resources to perform an accurate channel estimation. Moreover, in dense wireless networks, where many nodes share the same resources, an accurate channel estimation may not be feasible, even if the nodes have sufficient computing power.

Our work is along the lines of the work by Lozano, Heath, and Andrews \cite{lozano_fundamental_2013}. Indeed, Lozano et al.\ argued that cooperation can have severe limitations in the sense that the channel capacity for wireless networks can be bounded in the signal-to-noise (SNR). In other words, cooperation cannot convert an interference-limited network into a noise-limited one.  The main results in \cite{lozano_fundamental_2013} are  based on the analysis of a block-fading channel that models the channel within a cluster and takes out-of-cluster interference into account. Specifically, in \cite{lozano_fundamental_2013} two analyses are carried out:
\begin{itemize}
\item The analysis of a block-fading channel in a clustered system with pilot-assisted channel estimation and out-of-cluster interference that has a Gaussian distribution and grows linearly with the SNR. Due to this dependence of the interference on the SNR, the signal-to-interference-plus-noise ratio (SINR) is bounded in the SNR, resulting in the capacity being bounded in the SNR, too.
\item The analysis of a fully cooperative system, where all transmitters and all receivers cooperate, resulting effectively in MIMO transmission. It is assumed that the number of transmitters is greater than the number of time instants $L$ over which the block-fading channel stays constant. This precludes an accurate channel estimation. For this scenario, the maximum achievable rate is studied when the time-$k$ channel input is of the form $\sqrt{\mathsf{\SNR}}X_k$ (where the distribution of $X_k$ does not depend on the SNR), and it is demonstrated that this rate is bounded in the SNR.
\end{itemize}
For both analyses the conclusions are the same: in the absence of perfect CSI, the capacity of wireless networks is bounded in the SNR, hence, transmission over such networks is highly power-inefficient. This accentuates the limitation of cooperation in wireless networks.

However, the analyses by Lozano et al.\ are based on simplifying assumptions that weaken their conclusions. Indeed, in the former analysis of a clustered system with pilot-assisted channel estimation, the authors assume that the interference has a Gaussian distribution, which seems pessimistic. Indeed, it is well-known that the capacity of the additive-noise channel is minimized, among all power-constrained noise distributions, by \emph{identically and independently distributed (i.i.d.)} Gaussian noise \cite{blachman1957communication}, \cite{dobrushin1959optimum}. This result has been generalized to wireless networks \cite{shomorony_worst-case_2013}. Note, however, that Gaussian noise does not constitute the worst-case scenario when the input distribution is not Gaussian; see, e.g., \cite{shamai1992worst}. Based on the above observations, it may seem plausible that also for the considered channel model, Gaussian interference is the worst-case interference. Since the distribution of the interference depends on the distribution of the symbols transmitted by the interfering nodes, one may argue that the capacity may be unbounded if the interferers would use codebooks that give rise to non-Gaussian interference. 

As for the latter analysis of a fully cooperative system, here the authors assume that the time-$k$ channel inputs are of the form $\sqrt{\SNR} X_k$, where $X_k$ has a distribution that does not depend on the SNR. One may argue that the corresponding information rate is bounded in the SNR because of the suboptimal input distribution and not because of the limitations of cooperation. In fact, it has been demonstrated by Lapidoth and Moser in \cite[Th. 4.3]{lapidoth2003capacity} that for a memoryless channel and in the absence of perfect CSI such inputs give rise to a bounded information rate also in the point-to-point case. For noncoherent point-to-point memoryless fading channels, an input distribution that changes with the SNR is thus necessary in order to achieve an unbounded information rate.\footnote{However, in contrast to the case of perfect CSI, in its absence the capacity only grows double-logarithmically with the SNR \cite[Th. 4.2]{lapidoth2003capacity}.} Since the block-fading channel specializes to the memoryless fading channel when $L=1$, the observation that channel inputs of the form $\sqrt{\SNR}X_k$ yield a bounded information rate is perhaps not very surprising.

In this work, we explore whether the capacity of noncoherent wireless networks is still bounded in the SNR if we allow the input distribution to depend on the SNR. In contrast to the analysis by Lozano et al. \cite{lozano_fundamental_2013}, we assume for the sake of simplicity that the nodes do not cooperate. We do not know whether cooperation would give rise to an unbounded capacity. For simplicity, we further consider a flat-fading channel with an infinite number of interferers. The locations of these interferers enter the channel model through the variance of the fading coefficients corresponding to the paths between the interferers and the intended receiver. Without loss of generality, we order the interferers with respect to the variances of the corresponding fading coefficients: the fading coefficient of the first interferer has the largest variance, denoted by $\alpha_1$, the fading coefficient of the second interferer has the second-largest variance, denoted by $\alpha_2$, and so on. We consider a noncoherent scenario where transmitter and receiver are cognizant of the statistics of the fading coefficients, but are ignorant of their realization. We demonstrate that the observation by Lozano et al.\ continues to hold even if the input distribution is allowed to depend on the SNR, provided that the sequence of variances $\{\alpha_{\ell}\}$ decays at most exponentially and each interfering node transmits symbols that follow the same distribution as the symbols transmitted by the transmitting node. To this end, we derive an upper bound on the channel capacity that is bounded in the SNR.
 
The rest of this paper is organized as follows. Section II describes the channel model. Section III is devoted to channel capacity and the main result of our work. Section IV shows an upper bound on the channel capacity for the special case in which the variances of the channel coefficients decay exponentially. Section V contains the proof of our main result. Section VI concludes the paper with a summary and discussion of the results.

\section{Channel Model}\label{ChannelModel}
\begin{figure}[h]
\centering
\begin{subfigure}[b]{0.45\linewidth}
\centering
\includegraphics[width=0.65\textwidth]{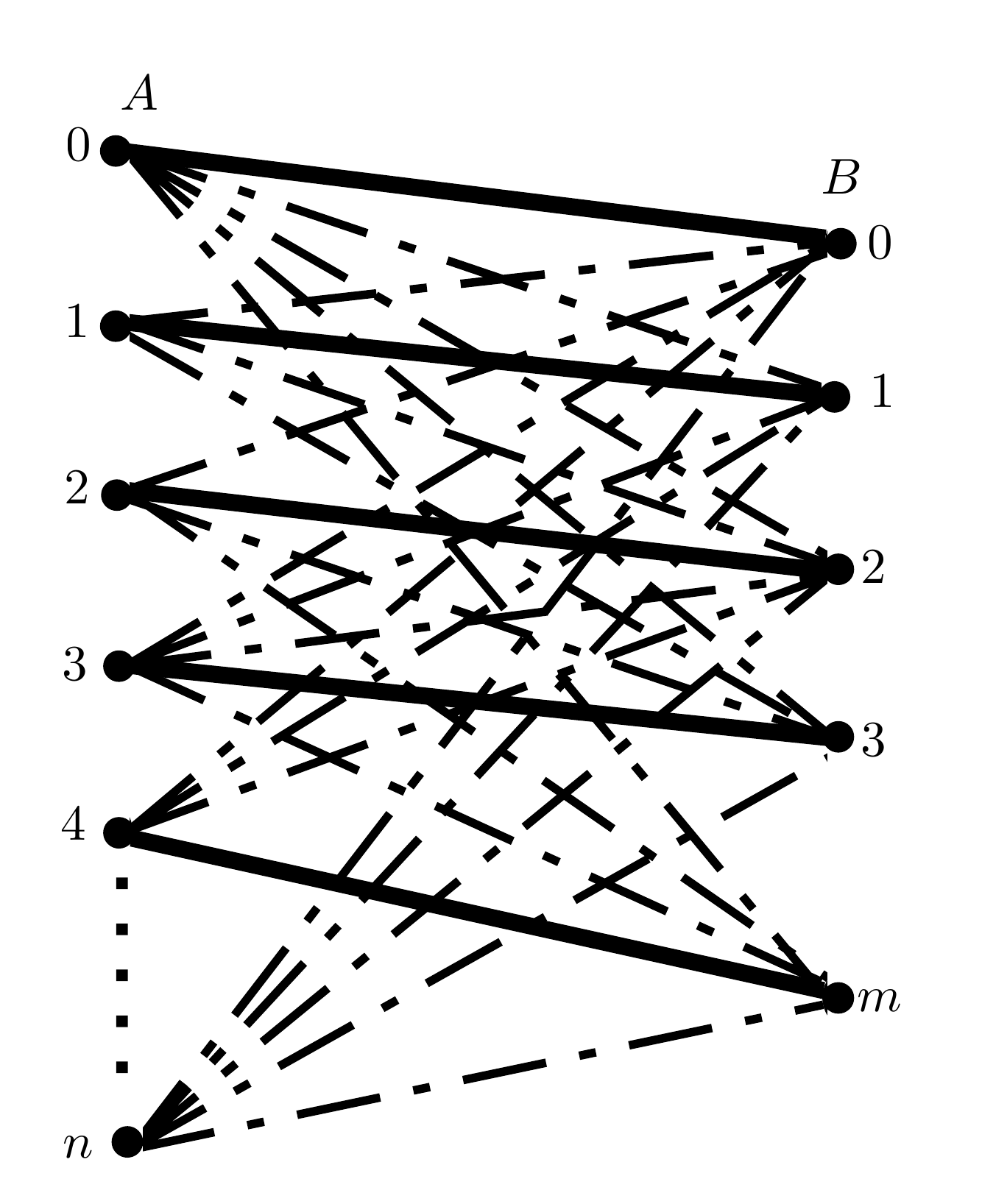}
\caption{Schema of a wireless network} \label{fig:Network_1}
\end{subfigure}
\begin{subfigure}[b]{0.45\linewidth}
\centering
\includegraphics[width=0.55\textwidth]{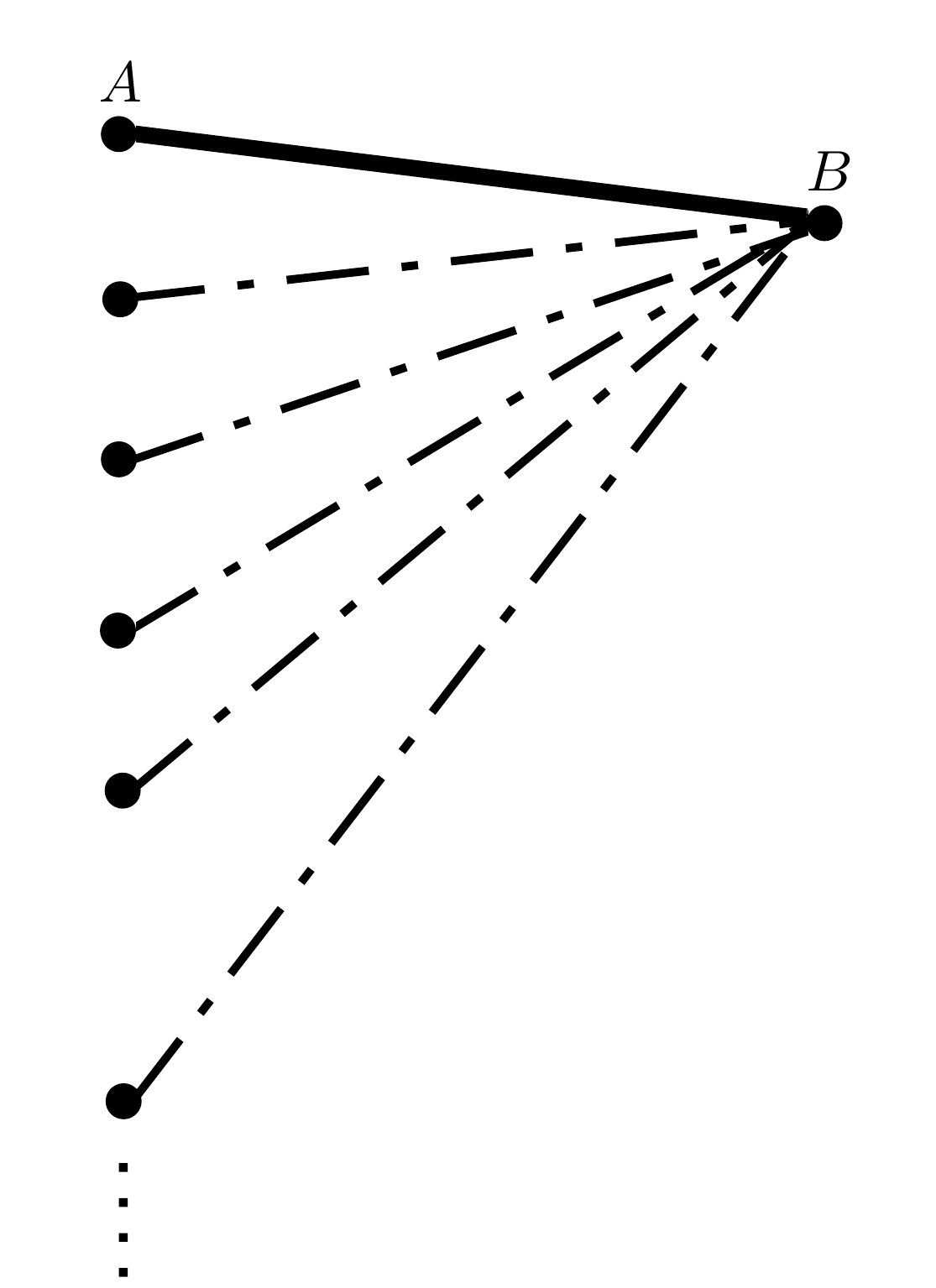}
\caption{Schema of a reduced wireless network} \label{fig:Network_2}
\end{subfigure}
\par\bigskip
\begin{subfigure}[b]{0.45\linewidth}
\centering
\includegraphics[width=0.8\textwidth]{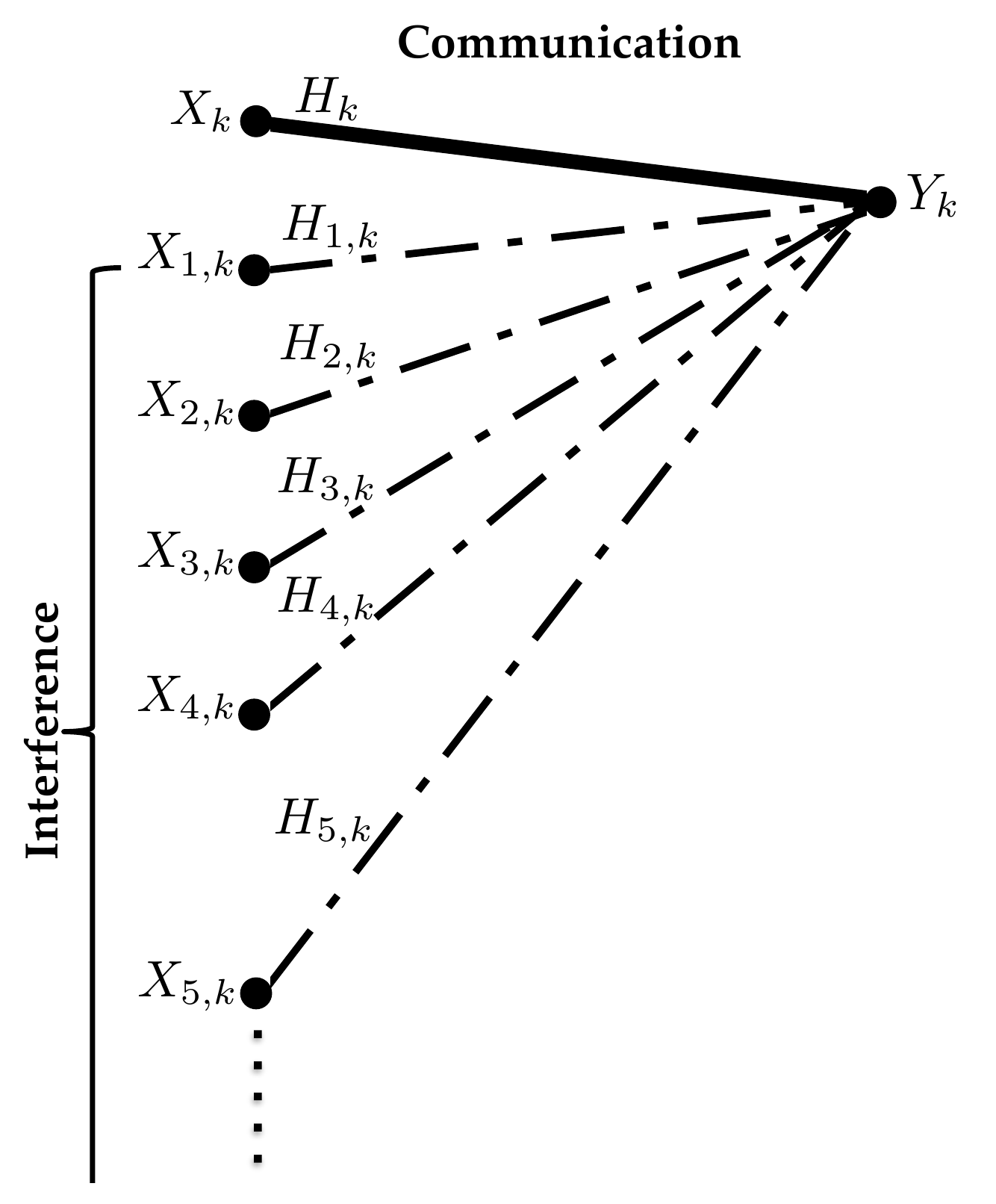}
\caption{Channel model} \label{fig:OurChannel}
\end{subfigure}
\caption{Wireless Networks}\label{fig:Networks}
\end{figure}

A network is composed of a number of users that are communicating with each other. A possible scenario is depicted in Figure~\ref{fig:Network_1}, where $n$ transmitting nodes (denoted in the figure by ``A") communicate with $m$ receiving nodes (denoted by ``B"). This network is sometimes referred to as an \emph{interference network}. For simplicity, we assume that the set of transmitting nodes and the set of receiving nodes are disjoint and both the transmitting nodes and the receiving nodes do not cooperate.

Since a characterization of all achievable rates in the network is unfeasible when the number of nodes is large, it is common to study the \emph{sum-rate capacity} of the network, which is defined as the supremum of the sum of all achievable rates in the network. However, it is \emph{prima facie} unclear whether the transmission strategy that achieves the sum-rate capacity will also be optimal in practice. Indeed, it may well be that the optimal transmission strategy consists of turning off all but one of the transmitting nodes, thereby minimizing the interference. Such a strategy prevents $n-1$ nodes from transmitting anything and is probably not very practical. In fact, practical constraints may demand that each node is offered roughly the same transmission rate or, at least, a transmission rate that is above a given threshold. In order to enforce such a solution, one could study the sum-rate capacity of the network under the constraint that all transmitting nodes transmit at the same rate, but obtaining an expression for such a capacity seems challenging. To model the problem in a way that it can be analyzed analytically, we simplify the original problem as follows:

Firstly, we consider the case where one transmitting node communicates with one receiving node and the interfering nodes emit symbols that interfer with this communication; see Figure (\ref{fig:Network_2}). To model a large network, we assume that there are an infinite number of interfering nodes. As performance measure we consider the capacity of the channel between the transmitting and receiving node. Secondly, to avoid transmission strategies for which the interfering nodes are turned off (which would, in fact, maximize the capacity), we assume that the symbols transmitted by the interfering nodes must follow the same distribution as the symbols transmitted by the transmitting node. By the channel coding theorem \cite{Cover}, this implies that each node (transmitting and interfering) is transmitting at the same rate, while at the same time it keeps the analysis accessible.

Note that, while the above simplifications permit a mathematical analysis of the channel capacity of the network, they are more restrictive than the analysis of the sum-rate capacity under the constraint that all nodes transmit at the same rate. Indeed, one can find transmission strategies for which each node transmits at the same rate, but the symbols produced by them do not follow the same distribution. One such strategy would be time-division multiple access (TDMA), where only one pair of nodes communicates with each other and the other nodes remain silent, and where the pairs of communicating nodes change after a given period until all nodes have transmitted their data. Clearly, in this case a node that transmits data produces symbols that follow a different distribution than the symbols emitted by silent nodes (for which the symbol transmitted at time $k$ is equal to zero). Thus, our constraint that all nodes shall produce symbols that follow the same distribution precludes, e.g., strategies such as TDMA. That said, the rates achievable with TDMA vanish as the number of users in the network tends to infinity, so in a large network TDMA seems not very attractive anyway.

We model the channel between the transmitting and receiving node by a discrete-time flat-fading channel whose complex-valued output  $Y_k$ at time $k\in\mathbb{Z}$ (where $\mathbb{Z}$ denotes the set of integers) corresponding to the time-$k$ channel input $X_k$ and the time-$k$ interfering symbols $X_{\ell,k}$, $\ell=1,2,\ldots$ is given by
\begin{equation}\label{ChModel}
Y_k=H_kX_k+\sum_{\ell=1}^{\infty}{H}_{\ell,k}{X}_{\ell,k} + Z_k, \quad k\in\mathbb{Z}.
\end{equation}
Here $Z_k$ models the time-$k$ additive noise; $H_k$ denotes the time-$k$ fading coefficient of the channel between the transmitter and receiver; and $H_{\ell,k}$, $\ell=1,2,\ldots$ denotes the time-$k$ fading coefficient of the link between the $\ell$-th interfering node and the receiver; see Figure (\ref{fig:OurChannel}) . We assume that the sequences $\{Z_k,\,k\in\Integers\}$, $\{H_k,\,k\in\Integers\}$, and $\{{H}_{\ell,k},\,k\in\Integers\}$, $\ell=1,2,\ldots$ are all sequences of i.i.d. complex random variables that they are independent of each other. We further assume that $Z_k\sim\mathcal{N_{\mathbb{C}}}(0,\sigma^2)$, $H_k \sim\mathcal{N_{\mathbb{C}}}(0,1)$ and ${H}_{\ell,k} \sim\mathcal{N_{\mathbb{C}}}(0,\alpha_\ell)$ for some $\alpha_{\ell}>0$, where we use the notation $U\sim\mathcal{N}_{\Complex}(\mu,\sigma^2)$ to indicate that $U$ is a circularly-symmetric, complex Gaussian random variable of mean $\mu$ and variance $\sigma^2$ ($\mathbb{C}$ denotes the set of complex numbers). We consider a noncoherent scenario where transmitter and receiver are cognizant of the statistics of the fading coefficients, but are ignorant of their realization.

We assume that the interfering nodes do neither cooperate with each other nor with the transmitting node. This implies that the sequences $\{X_k,\,k\in\Integers\}$ and $\{{X}_{\ell,k},\,k\in\Integers\}$, $\ell=1,2,\ldots$ are independent of each other. Since we also demand that each node produces symbols according to the same distribution, this implies that $\{X_k,\,k\in\Integers\}$ and $\{{X}_{\ell,k},\,k\in\Integers\}$, $\ell=1,2,\ldots$ are also identically distributed. Note, however, that this does not mean that the sequences of transmitted symbols are i.i.d.\ in time---to compute the capacity we will optimize over all distributions of $\{X_k,\,k\in\Integers\}$ that satisfy the power constraint.

Without loss of generality, we assume that the interfering nodes are ordered according to the variances of the corresponding fading coefficients, i.e., $\alpha_{\ell}\geq\alpha_{\ell'}$ for any $\ell<\ell'$.
We further assume that there exists a $0<\rho<1$ such that
\begin{equation}\label{EtaBounds_2}
\frac{\alpha_{\ell+1}}{\alpha_{\ell}}\geq \rho, \quad \ell=1,2,\ldots 
\end{equation}
The condition \eqref{EtaBounds_2} states that the variances of the fading coefficients corresponding to the interfering nodes decay at most exponentially. If we relate these variances to the path losses between the receiving node and the interfering nodes, then \eqref{EtaBounds_2} essentially requires that the distance between neighboring interfering nodes grows at most exponentially as we move farther away from the receiving node. We believe that this assumption is reasonably mild.

\section{Channel Capacity and Main Result}
We shall denote sequences such as $A_n, A_{n+1}, A_{n+2}, \cdots, A_m$ by $A_n^m$. We further denote by $\varlimsup$ the limit superior. We define the capacity of the above channel \eqref{ChModel} as\footnote{The logarithms used in this paper are natural logarithms. The capacity has thus the dimension ``nats per channel use''.}
\begin{equation}\label{Capacity-red}
C(\SNR)\triangleq \varlimsup_{n\rightarrow\infty}\frac{1}{n} \sup_{Q^n}I({X}_{1}^n;Y_1^n)
\end{equation}
where we assume that the sequences $X_1^n$ and $X_{\ell,1}^n$, $\ell=1,2,\ldots$ are independent of each other and that each such sequence has distribution $Q^n$. This is consistent with the simplifying assumptions introduced in Section~\ref{ChannelModel}. The supremum in \eqref{Capacity-red} is over all $n$-dimensional, complex-valued, probability distributions $Q^n$ satisfying the power constraint
\begin{eqnarray}\label{PowerC}
\int\frac{1}{n}\sum_{k=1}^{n}{|x_k|}^2dQ^n(x_1^n)\leq\mathsf{P}.
\end{eqnarray}
The $\mathrm{SNR}$ is defined as
\begin{equation}
\SNR \triangleq \frac{\mathsf{P}}{\sigma^2}.
\end{equation}
We do not claim that there is a coding theorem associated with \eqref{Capacity-red}, i.e., we do not claim that for any rate below $C(\SNR)$ there exists an encoding and decoding scheme for which the decoding error probability tends to zero as $n$ tends to infinity. (See \cite{Cover} for more details.) However, by Fano's inequality \cite[Sec. 7.9]{Cover}, we do know that for any rate above $C(\SNR)$ there exists no such encoding and decoding scheme. Thus, by demonstrating that $C(\SNR)$ is bounded in the $\SNR$, we also demonstrate that there exists no encoding and decoding scheme that has a rate that tends to infinity as $\SNR\to\infty$ and for which the decoding error probability vanishes as $n$ tends to infinity.

Our main result is the following.

\begin{theorem}[Main Result]
\label{MainTheorem}
For the channel model presented in Section~\ref{ChannelModel}, the capacity is bounded in the SNR, i.e.,
\begin{eqnarray}
\sup_{\SNR>0}C(\SNR)< \infty.
\end{eqnarray}
\end{theorem}
\begin{proof} See Section~\ref{Proof}.
\end{proof}

Thus, Theorem~\ref{MainTheorem} demonstrates that if the nodes do not cooperate, if they emit symbols that all follow the same distribution, and if the variances of the fading coefficients satisfy \eqref{EtaBounds_2}, then the capacity is bounded in the SNR, even if the distribution of the channel inputs is allowed to depend on the SNR. This result is more general than the result by Lozano et al.\cite{lozano_fundamental_2013} in the sense that we allow the distribution of the channel inputs to depend on the SNR, but it is less general than \cite{lozano_fundamental_2013} in the sense that we assume that the variances of the fading coefficients, i.e., $\{\alpha_{\ell}\}$, decay at most exponentially and that we do not allow cooperation between the nodes.

The condition \eqref{EtaBounds_2} is satisfied, for example, if
\begin{equation}
\label{AlphaRho}
\alpha_{\ell} = \rho^{\ell}, \quad \text{for some $0<\rho<1$.}
\end{equation}
In general, we can say that if $\{\alpha_{\ell}\}$ decays \emph{exponentially or more slowly}, then $C(\SNR)$ is bounded in the SNR. This result is reminiscent of a result obtained by Koch and Lapidoth \cite{koch_multipath_2010} that states that the capacity of single-user, frequency-selective, noncoherent fading channels is bounded in the SNR if the variances of the path gains decay exponentially or more slowly. Note, however, that the proof given in Section~\ref{Proof} is different from the one given in \cite{koch_multipath_2010}.

In order to prove the above result, we derive an upper bound on the capacity that does not depend on the SNR. Since this bound may also be of independent interest, we summarize it in the following proposition.

\begin{proposition}[Upper Bound] Consider the channel model introduced in Section~\ref{ChannelModel}. The capacity is upper-bounded by 
\begin{equation}\label{Proposition}
C(\SNR)\leq \log\frac{\pi }{2}+2\left(\frac{\eta_{\max}^2}{\eta_{\min}^2}-\frac{1}{2}\right)+\log \eta_{\max}^2+\frac{1}{2}\log\frac{\eta_{\max}^2}{\eta_{\min}^2},\quad \SNR>0.
\end{equation}
where $\eta^2_{\min}$ and $\eta^2_{\max}$ are defined as
\begin{eqnarray}
\eta_{\min}^2&\stackrel{\Delta}{=}&\min\left(1,\frac{\alpha_0}{\alpha_1}\right)\\
\eta_{\max}^2&\stackrel{\Delta}{=}&\max\left(\frac{\alpha_0}{\alpha_1}, \frac{1}{\rho}\right).
\end{eqnarray}
\end{proposition}
\begin{proof} Follows directly from the upper bound derived in the proof of Theorem~\ref{MainTheorem}; cf.\ (\ref{BoundedC}) and (\ref{K}).
\end{proof}
Assuming that $\{\alpha_{\ell}\}$ satisfies \eqref{AlphaRho}, we have $\eta_{\min}^2=1$ and $\eta_{\max}^2=\frac{1}{\rho}$. In this, the capacity is thus upper-bounded by
\begin{equation}\label{GeomDecay}
C(\SNR)\leq \log\frac{\pi }{2}+2\left(\frac{1}{\rho}-\frac{1}{2}\right)+\frac{3}{2}\log \frac{1}{\rho},\quad \SNR>0.
\end{equation}

\section{Exponential path loss}\label{ExponentialPL}
In the special case where the variances of the fading coefficients satisfy (\ref{AlphaRho}), i.e., where they decay exponentially, we obtain an upper bound that is tighter than the upper bound in (\ref{GeomDecay}). Indeed, it can be shown that in this case the capacity is upper-bounded by
\begin{eqnarray}\label{CRho}
C(\SNR)\leq\frac{1}{\rho}, \quad \SNR>0.
\end{eqnarray}
As we shall see, the proof of (\ref{CRho}) is much simpler than the proof of Theorem~\ref{MainTheorem}. 

To obtain (\ref{CRho}), we begin by deriving  an upper bound on the mutual information $I(X_1^n;Y_1^n)$ as follows. Since conditioning reduces entropy \cite[Th. 8.6.1]{Cover}, we have
\begin{IEEEeqnarray}{rCl}\label{MInf1}
I({X}_{1}^n;Y_1^n)&=& h(Y_1^n)-h(Y_1^n|X_1^n)\nonumber\\
&\leq& h(Y_1^{n})-h(Y_1^n|H_1^n,X_{1}^{n})\nonumber \\
&=& h(Y_1^n)-h(Y_1^n-{H}_{1}^n{X}_{1}^n) \IEEEeqnarraynumspace
\end{IEEEeqnarray}
where $Y_1^n-H_1^nX_1^n$ refers to sequence of symbols $Y_1-H_1 X_1,\ldots,Y_n-H_nX_n$. In the last step of (\ref{MInf1}), we use that the symbols emitted by the nodes are independent, so $X_1^n$ and $X_{\ell,1}^n$, $\ell=1,2,\ldots$ are independent.

We next define a new channel model, for which the channel output  $\bar{Y}_k$ at time $k$ $\in$ $\mathbb{Z}$ corresponding to the time-$k$ channel inputs ${X}_{\ell,1},\ldots,{X}_{\ell,k}$ is given by
\begin{eqnarray}\label{barY_k}
\bar{Y}_k&=&\sum_{\ell=0}^{\infty}{H}_{\ell+1,k}{X}_{\ell,k}+{Z}_k,\quad k \in\mathbb{Z}.
\end{eqnarray}
Observe that $(H_{\ell+1,1},\ldots,H_{\ell+1,n})$ can be obtained from $(H_{\ell,1},\ldots,H_{\ell,n})$ by changing $\ell$ to $\ell+1$. It thus follows from (\ref{AlphaRho}) and the independence of the fading processes that $(H_{\ell+1,1},\ldots,H_{\ell+1,n})$, $\ell=0,1,\ldots$ has the same joint distribution as $(\sqrt{\rho}H_{\ell,1},\ldots,\sqrt{\rho}H_{\ell,n})$, $\ell=0,1,\ldots$ Similarly, the additive Gaussian noise $(Z_1,\ldots,Z_n)$ has the same distribution as $(\sqrt{\rho}Z_1^{(1)}+\sqrt{1-\rho}Z_1^{(2)},\ldots,\sqrt{\rho}Z_n^{(1)})+\sqrt{1-\rho}Z_n^{(2)})]$, where $Z_1^{(1)},\ldots,Z_n^{(1)},Z_1^{(2)},\ldots,Z_n^{(2)}$ are i.i.d., circularly-symmetric, complex Gaussian random variable of zero-mean and variance $\sigma^2$. It follows by comparing (\ref{barY_k}) with (\ref{ChModel}) that
\begin{eqnarray}\label{barY_knew}
\bar{Y}_1^n\eqlaw \sqrt{\rho}Y_1^n+\sqrt{1-\rho}{Z_1^n}^{(2)}
\end{eqnarray}
where we use $\eqlaw$ to denote equivalence in the probability law. With this, we can compute the entropy of the second term in (\ref{MInf1}) as 
\begin{eqnarray}\label{h(barY_1^n)}
h(\bar{Y}_1^n)&=& h\left(\sqrt{\rho}Y_1^n+\sqrt{1-\rho}{Z_1^n}^{(2)}\right)\nonumber\\
&\geq& h\left(\sqrt{\rho}Y_1^n+\sqrt{1-\rho}{Z_1^n}^{(2)}\big|{Z_1^n}^{(2)}\right)\nonumber\\
&=& n\log\rho+h(Y_1^n)
\end{eqnarray}
where the second step follows from the behavior of differential entropy under translation \cite[Th.8.6.3]{Cover} and scaling \cite[Th.8.6.4]{Cover}, and because conditioning reduces entropy; the last step follows because $Y_1^n$ is independent of ${Z_1^n}^{(2)}$.

Replacing (\ref{h(barY_1^n)}) in (\ref{MInf1}), we obtain

\begin{eqnarray}\label{MInfRho}
\frac{1}{n}I(X_1^n;Y_1^n)&\leq&\frac{1}{n}\left( h(Y_1^n)-n\log\rho-h(Y_1^n)\right)\nonumber\\
&=& \log\frac{1}{\rho}.
\end{eqnarray}
This proves (\ref{CRho}).

\section{Proof of Theorem~\ref{MainTheorem}}\label{Proof}
To prove Theorem~\ref{MainTheorem}, we derive an upper bound on the channel capacity that is independent of the SNR. We begin with the same steps as in Section~\ref{ExponentialPL} to obtain (\ref{MInf1}). Thus, we derive an upper bound on the mutual information $I(X_1^n;Y_1^n)$ as follows
\begin{IEEEeqnarray}{rCl}\label{MInf}
I({X}_{1}^n;Y_1^n)&=& h(Y_1^n)-h(Y_1^n|X_1^n)\nonumber\\
&\leq& h(Y_1^{n})-h(Y_1^n|H_1^n,X_{1}^{n})\nonumber \\
&=& h(Y_1^n)-h(Y_1^n-{H}_{1}^n{X}_{1}^n). \IEEEeqnarraynumspace
\end{IEEEeqnarray}
For the first entropy on the right-hand side (RHS) of (\ref{MInf}) we rewrite the channel model in (\ref{ChModel}) as 
\begin{equation}\label{Y_k}
Y_k=\sum_{\ell=0}^{\infty}{H}_{\ell,k}{X}_{\ell,k} + Z_k,\quad k \in \mathbb{Z}
\end{equation}
where we define $H_{0,k}\triangleq H_k$. For the second entropy in (\ref{MInf}) we define again a new channel model, as we have done in the last section. Specifically, the time-$k$ channel output  $\tilde{Y}_k$ of the new channel corresponding to the time-$k$ channel inputs ${X}_{\ell,1},\ldots,{X}_{\ell,k}$ is given by
\begin{eqnarray}\label{^Y_k}
\tilde{Y}_k=\sum_{\ell=0}^{\infty}\tilde{H}_{\ell+1,k}{X}_{\ell,k}+\tilde{Z}_k,\quad k \in \mathbb{Z}  
\end{eqnarray}
where the fading coefficients $\{\tilde{H}_{\ell,k}\}$  have the same distribution as $\{{H}_{\ell,k}\}$ but are independent of $\{H_{\ell,k}\}$ and, likewise, the additive-noise terms $\{\tilde{Z}_k\}$ have the same distribution as $\{Z_k\}$ but are independent of $\{Z_k\}$. It follows that the sequences $\{\tilde{Z}_k\}$ and $\{\tilde{H}_{\ell,k}\}$ are i.i.d. and independent of each other, with  $\tilde{Z}_k$ $\sim\mathcal{N_{\mathbb{C}}}(0,{\sigma}^2)$ and $\tilde{H}_{\ell,k}\sim\mathcal{N_{\mathbb{C}}}(0,\alpha_\ell)$. Since $\{X_k\}$ and $\{X_{\ell,k}\}$, $\ell=1,2,\ldots$ have the same distribution, it follows that $\tilde{Y}_1^n$ has the same distribution as $Y_1^n-H_1^nX_1^n$, hence $h(Y_1^n-H_1^nX_1^n)=h(\tilde{Y}_1^n)$.

To find an upper bound on (\ref{MInf}) we use the identity $h(A)-h(B)=h(A|B)-h(B|A)$ to obtain
\begin{IEEEeqnarray}{lCl}\label{MInf_2}
h(Y_1^n)-h(\tilde{Y}_1^n)&=& h(Y_1^n|\tilde{Y}_1^n)-h(\tilde{Y}_1^n|Y_1^n) \nonumber \\
&\leq&\sum_{k=1}^{n}h(Y_k|\tilde{Y}_k)-\sum_{k=1}^{n}h(\tilde{Y}_k|\tilde{Y}_1^{k-1},Y_1^n)\nonumber	\\
&=& \sum_{k=1}^{n}\Bigg( h\bigg(\tilde{Y_k}\frac{Y_k}{\tilde{Y_k}}\bigg|\tilde{Y}_k\bigg)-h(\tilde{Y_k}|\tilde{Y}_1^{k-1},Y_1^n)\Bigg)\nonumber\\
&\leq& \sum_{k=1}^{n}\Bigg( \E{\log|\tilde{Y}_k|^2}+{h\bigg(\frac{Y_k}{\tilde{Y}_k}\bigg)}-h(\tilde{Y_k}|\tilde{Y}_1^{k-1},Y_1^n)\Bigg)
\end{IEEEeqnarray}
where the second step follows from the chain rule for entropy and because conditioning reduces entropy; the third step follows because if we multiply and divide by $\tilde{Y_k}$ we do not change the conditional differential entropy; the last inequality follows from the behavior of differential entropy under scaling by a complex random variable and because conditioning reduces entropy.

In order to find an upper bound on $h(Y_1^n)-h(\tilde{Y}_1^n)$, we calculate bounds for each term in (\ref{MInf_2}). For the first term we have
\begin{IEEEeqnarray}{rCl}\label{UB_1}
\E{\log|\tilde{Y}_k|^2 }&=&\E{\E{\log\left(|\tilde{Y}_k|^2\right)\Big|\{{X}_{\ell,k}\}}}\nonumber\\
&\leq& \E{\log\left(\E{|\tilde{Y}_k|^2\Big|\{{X}_{\ell,k}\}}\right)}
\end{IEEEeqnarray}
where last step follows from Jensen's inequality.

For the third term on the RHS of (\ref{MInf_2}), we use that conditioning reduces entropy to obtain
\begin{IEEEeqnarray}{lCl}\label{UB_3}
h\left(\tilde{Y_k}|\tilde{Y}_1^{k-1},Y_1^n\right)&\geq& h\left(\tilde{Y_k}\Big|\tilde{Y}_1^{k-1},Y_1^n,\{{X}_{\ell,k}\}\right)\nonumber\\
&=& h\left(\tilde{Y_k}\Big|\{{X}_{\ell,k}\}\right)\nonumber \\
&=&\log(\pi e)+\E{\log\bigg(\E{|\tilde{Y}_k|^2\Big|\{{X}_{\ell,k}\}}\bigg)}
\end{IEEEeqnarray}
where the first equality follows because $\{\tilde{Z}_k\}$ and $\{\tilde{H}_{\ell,k}\}$ are i.i.d. and independent of $(\{H_{\ell,k}\}, \{Z_k\})$, so conditioned on $\{{X}_{\ell,k}\}$, $\tilde{Y}_k$ is independent of $\left(\tilde{Y}_1^{k-1}, Y_1^n\right)$; the second  equality follows because, conditioned on the inputs $\{{X}_{\ell,k}\}$, $\tilde{Y}_k$ has a Gaussian distribution.

Now we bound the second term in (\ref{MInf_2}). To this end, we first note that, conditioned on the inputs $\{X_{\ell,k}\}$, the ratio $\frac{Y_k}{\tilde{Y}_k} $ is a ratio of independent, circularly-symmetric complex random variables. This implies that $\frac{Y_k}{\tilde{Y}_k} $ is also circularly symmetric unconditioned on $\{X_{\ell,k}\}$. Indeed, express  $Y_k$ and $\tilde{Y}_k$ as
\begin{eqnarray}\label{Yk-Yk'}
Y_k&=& |Y_k|e^{j\phi_Y} \\
\tilde{Y}_k&=& |\tilde{Y}_k|e^{j\phi_{\tilde{Y}}}
\end{eqnarray}
for some $\phi_Y$,$\phi_{\tilde{Y}} \in \left[0,2\pi\right)$. By circular symmetry, it follows that 
\begin{eqnarray}
\frac{Y_k}{\tilde{Y}_k}&=&\left|\frac{Y_k}{\tilde{Y}_k}\right|e^{j(\phi_Y-\phi_{\tilde{Y}})}
\end{eqnarray}
where, conditioned on $\{X_{\ell,k}\}$, $\phi_Y$ and $\phi_{\tilde{Y}}$ are uniformly distributed and independent of $\left|\frac{Y_k}{\tilde{Y}_k}\right|$. Since this is true irrespective of $\{X_{\ell,k}\}$, it follows that $\frac{Y_k}{\tilde{Y}_k}$ is circularly symmetric.

Once we have determined that the ratio $\frac{Y_k}{\tilde{Y}_k}$ is circularly symmetric, we can apply \cite[Lemma 6.16]{lapidoth2003capacity} to obtain

\begin{IEEEeqnarray} {lCl}\label{h(W)AS}
h\bigg(\frac{Y_k}{\tilde{Y}_k}\bigg)&=&\log2\pi+h\bigg(\bigg|\frac{Y_k}{\tilde{Y}_k}\bigg|\bigg)+\E{\log\bigg(\bigg|\frac{Y_k}{\tilde{Y}_k}\bigg|\bigg)}
\nonumber\\
&\leq& \log2\pi+h\bigg(\bigg|\frac{Y_k}{\tilde{Y}_k}\bigg|\bigg)+\log\left(\E{\bigg|\frac{Y_k}{\tilde{Y}_k}\bigg|}\right)
\end{IEEEeqnarray}
where the inequality follows from Jensen's inequality.

Let $R\triangleq\left|\frac{Y_k}{\tilde{Y}_k}\right|$. Conditioned on $\{X_{\ell,k}\}$, $R$ is the ratio of two independent Rayleigh distributed random variables, whose probability density function is given by \cite[p. 68, Eq. (7.44)]{Simon2006PDI1212190}
\begin{IEEEeqnarray}{lCl}\label{pdf}
f_{R|\{X_{\ell,k}\}}(r|\{x_{\ell,k}\})&=&\frac{2\E{|\tilde{Y}_k|^2\Big|\{X_{\ell,k}\}=\{x_{\ell,k}\}}\E{|{Y}_k|^2\Big|\{X_{\ell,k}\}=\{x_{\ell,k}\}}r}{\left(\E{|\tilde{Y}_k|^2\Big|\{X_{\ell,k}\}=\{x_{\ell,k}\}}r^2+\E{|{Y}_k|^2\Big|\{X_{\ell,k}\}=\{x_{\ell,k}\}}\right)^2},\nonumber\\
\end{IEEEeqnarray}
for $r>0$.

To compute the needed differential entropy of $R$, we need to marginalize over $\{X_{\ell,k}\}$. To this end, we define
\begin{eqnarray}\label{eta_def}
\eta^2(\{x_{\ell,k}\})\stackrel{\Delta}{=}\frac{\E{|Y_k|^2\Big|\{X_{\ell,k}\}=\{x_{\ell,k}\}}}{\E{|\tilde{Y}_k|^2\Big|\{X_{\ell,k}\}=\{x_{\ell,k}\}}}
\end{eqnarray}
which can be evaluated  as 
\begin{IEEEeqnarray}{lCl}\label{eta_eval}
\eta^2(\{x_{\ell,k}\})=\frac{\sum_{\ell=0}^{\infty}\alpha_{\ell}|{x}_{\ell,k}|^2+\sigma^2}{\sum_{\ell=0}^{\infty}\alpha_{\ell+1}|{x}_{\ell,k}|^2+\sigma^2}.
\end{IEEEeqnarray}
With this, we obtain
\begin{eqnarray}\label{f_R}
f_R(r)=2r\E{\frac{\eta^2(\{X_{\ell,k}\})}{\left(r^2+\eta^2(\{X_{\ell,k}\})\right)^2}}
\end{eqnarray}
where the expectation is over $\{X_{\ell,k}\}$. This yields for the differential entropy of $R$
\begin{IEEEeqnarray}{lCl}\label{ratio_1}
h\bigg(\bigg|\frac{Y_k}{\tilde{Y}_k}\bigg|\bigg)&\stackrel{\Delta}{=}&-\int_{0}^{\infty}f_R(r)\log f_R(r)dr\nonumber\\
&=&{}-
\int_{0}^{\infty}2r\E{\frac{\eta^2{(\{X_{\ell,k}\})}}{\big(r^2+\eta^2{(\{X_{\ell,k}\})}\big)^2}}\log \left(2r\E{\frac{\eta^2{(\{X_{\ell,k}\})}}{\big(r^2+\eta^2{(\{X_{\ell,k}\})}\big)^2}}\right)dr\nonumber\\
&=&{}- \int_{0}^{\infty}2r\E{\frac{\eta^2{(\{X_{\ell,k}\})}}{\big(r^2+\eta^2{(\{X_{\ell,k}\})}\big)^2}}\log (2r) dr\nonumber\\
&&\quad{}- \int_{0}^{\infty}2r\E{\frac{\eta^2{(\{X_{\ell,k}\})}}{\big(r^2+\eta^2{(\{X_{\ell,k}\})}\big)^2}}\log \left(\E{\frac{\eta^2{(\{X_{\ell,k}\})}}{\big(r^2+\eta^2{(\{X_{\ell,k}\})}\big)^2}}\right) dr.
\end{IEEEeqnarray}
Computing the first integral yields
\begin{IEEEeqnarray}{lCl}\label{Integral1}
- \int_{0}^{\infty}2r\E{\frac{\eta^2{(\{X_{\ell,k}\})}}{\big(r^2+\eta^2{(\{X_{\ell,k}\})}\big)^2}}\log (2r) dr
&=& -\E{\int_{0}^{\infty}2r\frac{\eta^2{(\{X_{\ell,k}\})}}{\big(r^2+\eta^2{(\{X_{\ell,k}\})}\big)^2}\log(2r)dr}\nonumber\\
&=&
-\log(2)-\E{\int_{0}^{\infty}2r\frac{\eta^2{(\{X_{\ell,k}\})}}{\big(r^2+\eta^2{(\{X_{\ell,k}\})}\big)^2}\log(r)dr}\nonumber\\
&=&-\log(2) -\frac{1}{2}\E{\log (\eta^2{(\{X_{\ell,k}\})})}
\end{IEEEeqnarray}
where the first step follows by Fubini's theorem \cite[p. 108, Th. 2.6.6]{ash2000probability}, the last step follows from computing the integral using the change of variable $t= r^2$ and \cite[p. 535, Sec. 4.23, Eq. (5)]{Gradshteyn2007}.

To bound the second integral in (\ref{ratio_1}), we require bounds on $\eta^2{(\{x_{\ell,k}\})}$, which we shall present in the following lemma.

\begin{lemma} \label{Lemma}
Assume that $\{\alpha_{\ell}\}$ satisfies (\ref{EtaBounds_2}), namely, $\frac{\alpha_{\ell+1}}{\alpha_{\ell}}\geq\rho$, $\ell=1,2,\ldots$ Then, $\eta^2{(\{x_{\ell,k}\})}$ is bounded by
\begin{eqnarray}
\eta_{\min}^2\leq\eta^2\left(\{x_{\ell,k}\}\right)\leq\eta_{\max}^2
\end{eqnarray}
where
\begin{eqnarray}
\eta^2_{\min}&\stackrel{\Delta}{=}&\min\left(1,\frac{\alpha_0}{\alpha_1}\right)\\
\eta^2_{\max}&\stackrel{\Delta}{=}&\max\left(\frac{\alpha_0}{\alpha_1},\frac{1}{\rho}\right).
\end{eqnarray}
\end{lemma}
\begin{proof}
By \eqref{eta_eval}, we have
\begin{IEEEeqnarray}{lCr}\label{EtaBounds_1}
{\eta^2{(\{x_{\ell,k}\})}}={\frac{\sum_{\ell=0}^{\infty} \alpha_{\ell}|{x}_{\ell,k}|^2+\sigma^2}{\sum_{\ell=0}^{\infty} \alpha_{\ell+1}|{x}_{\ell,k}|^2+\sigma^2}}.
\end{IEEEeqnarray}
We next maximize and minimize 
\begin{equation}\label{gx}
\eta^2{(\{x_{\ell,k}\})}=\frac{\alpha_0|{x}_{0,k}|^2+\alpha_1|{x}_{1,k}|^2+\cdots+\sigma^2}{\alpha_1|{x}_{0,k}|^2+\alpha_2|{x}_{1,k}|^2+\cdots+\sigma^2}
\end{equation}
over $x_{0,k}, x_{1,k},\ldots$ To this end, we first fix $x_{1,k},x_{2,k},\ldots$ and optimize over $x_{0,k}$. Thus, we write $\eta^2{(\{x_{\ell,k}\})}$ as
\begin{equation}\label{fx}
f(x)\stackrel{\Delta}{=}\frac{\alpha_0 x+\beta_1}{\alpha_1 x+\beta_2} 
\end{equation}
where $x$ corresponds to $|{x}_{0,k}|^2$, and $\beta_1$ and $\beta_2$ are the remaining terms in the numerator and denominator in \eqref{gx}, respectively, i.e., $\beta_1=\alpha_1|x_{1,k}|^2+\cdots+\sigma^2$ and $\beta_2=\alpha_2|x_{1,k}|^2+\cdots+\sigma^2$.  It is easy to show that
\begin{eqnarray}
\min\left(\frac{\alpha_0}{\alpha_1},\frac{\beta_1}{\beta_2}\right)\leq f(x) \leq \max \left(\frac{\alpha_0}{\alpha_1},\frac{\beta_1}{\beta_2}\right).
\end{eqnarray}
We next express $\frac{\beta_1}{\beta_2}$ as
\begin{equation}\label{fx^}
\tilde{f}(x)=\frac{\alpha_1 x+\tilde{\beta}_1}{\alpha_2 x+\tilde{\beta}_2} 
\end{equation}
where $x$ corresponds to $|{x}_{1,k}|^2$, and $\tilde{\beta}_1$ and $\tilde{\beta}_2$ are the remaining terms in the numerator and denominator, respectively, i.e., $\tilde{\beta}_1=\alpha_2|{x}_{2,k}|^2+\ldots+\sigma^2$ and $\tilde{\beta}_2=\alpha_3|{x}_{2,k}|^2+\ldots+\sigma^2 $. 
Again,  $\min\left(\frac{\alpha_1}{\alpha_2},\frac{\tilde{\beta}_1}{\tilde{\beta}_2}\right)\leq \tilde{f}(x) \leq \max \left(\frac{\alpha_1}{\alpha_2},\frac{\tilde{\beta}_1}{\tilde{\beta}_2}\right)$, so
\begin{eqnarray}
\min\left(\frac{\alpha_0}{\alpha_1},\frac{\alpha_1}{\alpha_2},\frac{\tilde{\beta}_1}{\tilde{\beta}_2}\right)\leq f(x) \leq \max \left(\frac{\alpha_0}{\alpha_1},\frac{\alpha_1}{\alpha_2},\frac{\tilde{\beta}_1}{\tilde{\beta}_2}\right).
\end{eqnarray}
Repeating these steps and considering that at the end we have $\frac{\beta_1}{\beta_2}=\frac{\sigma}{\sigma}=1$, we obtain
\begin{IEEEeqnarray}{rCL}\label{EtaBounds_3}
\eta^2{(\{x_{\ell,k}\}}&\leq& \max\bigg(\sup_{k=0,1,\ldots}\bigg(\frac{\alpha_k}{\alpha_{k+1}}\bigg),1 \bigg)\stackrel{\Delta}{=}\eta_{\max}^2,\\
\eta^2{(\{x_{\ell,k}\})}&\geq& \min\bigg(1,\inf_{k=0,1,\ldots}\bigg(\frac{\alpha_{k}}{\alpha_{k+1}} \bigg)\bigg)\stackrel{\Delta}{=}\eta_{\min}^2.
\end{IEEEeqnarray}
Using the assumption that $\alpha_1\geq\alpha_2\geq\ldots$ and (\ref{EtaBounds_2}), we have
\begin{eqnarray}\label{a_k/a_k+1}
1\leq\frac{\alpha_k}{\alpha_{k+1}}\leq \frac{1}{\rho}, \qquad k=1,2,\ldots
\end{eqnarray}
from which we obtain
\begin{eqnarray}\label{eta_min/max}
\eta_{\min}^2&=&\min\left(1,\frac{\alpha_0}{\alpha_1}\right)\\
\eta_{\max}^2&=&\max\left(\frac{\alpha_0}{\alpha_1},\frac{1}{\rho}\right).
\end{eqnarray}
\end{proof}
By the continuity of $x\rightarrow \frac{x}{\left(r^2+x\right)^2}$, it follows that we can define
\begin{eqnarray}\label{Eta*}
\eta_{*}^2=\argmin_{\eta_{\min}^2\leq\eta^2\leq \eta_{\max}^2}\frac{\eta^2}{\left(r^2+\eta^2\right)^2}
\end{eqnarray}
which satisfies $\eta_{\min}^2\leq\eta_*^2\leq\eta_{\max}^2$.

We next apply (\ref{Eta*}) to compute an upper bound on the second integral in (\ref{ratio_1}). By definition, we have that
\begin{equation}\label{Ex>x}
\E{\left(\frac{\eta^2(\{X_{\ell,k}\})}{(r^2+\eta^2(\{X_{\ell,k}\}))^2}\right)}\geq\frac{\eta_*^2}{(r^2+\eta_*^2)^2}
\end{equation}
so by the monotonicity of the logarithm function
\begin{IEEEeqnarray}{lll}\label{Integral2}
- \int_{0}^{\infty}2r\E{\frac{\eta^2{(\{X_{\ell,k}\})}}{\big(r^2+\eta^2{(\{X_{\ell,k}\})}\big)^2}}\log \left(\E{\frac{\eta^2{(\{X_{\ell,k}\})}}{\big(r^2+\eta^2{(\{X_{\ell,k}\})}\big)^2}}\right) dr \nonumber\\
\qquad\qquad\leq -\int_{0}^{\infty}\E{\frac{2r\eta^2{(\{X_{\ell,k}\})}}{\left(r^2+\eta^2{(\{X_{\ell,k}\})}\right)^2}}\log\left(\frac{\eta_{*}^2}{\left(r^2+\eta_{*}^2\right)^2}\right)dr\nonumber\\
\qquad\qquad=-\E{\int_{0}^{\infty}\frac{2r\eta^2{(\{X_{\ell,k}\})}}{\left(r^2+\eta^2{(\{X_{\ell,k}\})}\right)^2}\log\left(\frac{\eta_{*}^2}{\left(r^2+\eta_{*}^2\right)^2}\right)dr}\nonumber\\
\qquad\qquad=-\log\left(\eta_{*}^2\right)
+\E{\int_{0}^{\infty}\frac{4r\eta^2{(\{X_{\ell,k}\})}}{\left(r^2+\eta^2{(\{X_{\ell,k}\})}\right)^2}\log\left(r^2+\eta_{*}^2\right)dr}\IEEEeqnarraynumspace
\end{IEEEeqnarray}
where the second step follows by Fubini's theorem.

To derive an upper bound on the second term in (\ref{Integral2}), we integrate by parts. To this end, we have to distinguish between the cases $\eta_{*}^2\geq\eta^2{(\{x_{\ell,k}\})}$ and $\eta_{*}^2<\eta^2{(\{x_{\ell,k}\})}$. The details are carried in Appendix \ref{Ap_Integralbyparts}, the results are as follows:

If $\eta_{*}^2\geq\eta^2{(\{x_{\ell,k}\})}$, then
\begin{IEEEeqnarray}{lCl}\label{eta*>eta}
\E{\int_{0}^{\infty}\frac{4r\eta^2{(\{X_{\ell,k}\})}\log\left(r^2+\eta_{*}^2\right)}{\left(r^2+\eta^2{(\{X_{\ell,k}\})}\right)^2}dr\mathds{1}\{\eta_{*}^2\geq\eta^2{(\{X_{\ell,k}\})}\} }&\leq&  \E{\left(2\log\left(\eta_{*}^2\right)+2\right)\mathds{1}\{\eta_{*}^2\geq\eta^2{(\{X_{\ell,k}\})}\}}\nonumber\\
&=&\left(2\log\left(\eta_{*}^2\right)+2\right)\mathsf{Pr}\left(\eta_{*}^2\geq\eta^2{(\{X_{\ell,k}\})}\right)
\end{IEEEeqnarray}
where the last step follows because $\eta_{*}^2$ is deterministic. Here $\mathds{1}\{\cdot\}$ denotes the indicator function.

If $\eta_{*}^2<\eta^2{(\{x_{\ell,k}\})}$, then
\begin{IEEEeqnarray}{lCl}\label{eta*<eta}
\E{\int_{0}^{\infty}\frac{4r\eta^2{(\{X_{\ell,k}\})}\log\left(r^2+\eta_{*}^2\right)}{\left(r^2+\eta^2{(\{X_{\ell,k}\})}\right)^2}dr\mathds{1}\{\eta_{*}^2<\eta^2{(\{X_{\ell,k}\})}\}}&\leq& \E{\left(2\log(\eta_{*}^2)+2\frac{\eta^2(\{X_{\ell,k}\})}{\eta_{*}^2}\right)\mathds{1}\{\eta_{*}^2<\eta^2{(\{X_{\ell,k}\})}\}}\nonumber\\
&\leq&\left(2\log(\eta_{*}^2)+2\frac{\eta_{\max}^2}{\eta_*^2}\right)\mathsf{Pr}\left(\eta_{*}^2<\eta^2{(\{X_{\ell,k}\})}\right) 
\end{IEEEeqnarray}
where the last step follows because $\eta^2{(\{x_{\ell,k}\})}\leq\eta^2_{\max}$ and because $\log(\eta_{*}^2)$ and $\eta_{\max}^2$ are deterministic.

Combining (\ref{Integral1}) and (\ref{Integral2})-(\ref{eta*<eta}) with (\ref{ratio_1}), we obtain
\begin{IEEEeqnarray}{lCl}\label{ratio_2}
h\left(\left|\frac{Y_k}{\tilde{Y}_k}\right|\right)&\leq& -\log(2)-\frac{1}{2}\E{\log(\eta^2{(\{X_{\ell,k}\})})}-\log(\eta_*^2)+\left(2\log\left(\eta_{*}^2\right)+2\right)\left(\mathsf{Pr}\left(\eta_{*}^2\geq\eta^2{(\{X_{\ell,k}\})}\right)\right)\nonumber\\
&&{}+\left(2\log(\eta_{*}^2)+2\frac{\eta_{\max}^2}{\eta_*^2}\right)\mathsf{Pr}\left(\eta_{*}^2<\eta^2{(\{X_{\ell,k}\})}\right)\nonumber\\  
&\leq&-\log(2)-\frac{1}{2}\E{\log(\eta^2(\{X_{\ell,k}\})}+\log\left(\eta_{*}^2\right)
+2\frac{\eta_{\max}^2}{\eta_{*}^2}\nonumber\\
&\leq&-\log(2)-\frac{1}{2}\log(\eta_{\min}^2)+\log\left(\eta_{*}^2\right)+2\frac{\eta_{\max}^2}{\eta_{*}^2}\IEEEeqnarraynumspace
\end{IEEEeqnarray}
where the second step follows by upper-bounding $1\leq\frac{\eta_{\max}^2}{\eta_{*}^2}$ and the last step follows by lower-bounding $\E{\log(\eta^2{(\{X_{\ell,k}\})})}\geq\log(\eta_{\min}^2)$.

Together with (\ref{h(W)AS}), this yields
\begin{IEEEeqnarray}{lCl}\label{UB_2}
h\bigg(\frac{Y_k}{\tilde{Y}_k}\bigg)
&\leq& \log2\pi+h\bigg(\bigg|\frac{Y_k}{\tilde{Y}_k}\bigg|\bigg)+{\log\Bigg(\E{\bigg|\frac{Y_k}{\tilde{Y}_k}\bigg|}\Bigg)}\nonumber\\
&\leq& \log(\pi)-\frac{1}{2}\log(\eta_{\min}^2)+\log\left(\eta_{*}^2\right)
+2\frac{\eta_{\max}^2}{\eta_{*}^2}+\log\Bigg(\E{\bigg|\frac{Y_k}{\tilde{Y}_k}\bigg|}\Bigg)\nonumber\\
&=&\log(\pi)-\frac{1}{2}\log(\eta_{\min}^2)+\log\left(\eta_{*}^2\right)
+2\frac{\eta_{\max}^2}{\eta_{*}^2}+{\log\left(\E{\frac{\pi}{2}\eta{(\{X_{\ell,k}\})}\right)}}\nonumber\\ 
&=& 2\log(\pi)-\log(2)-\frac{1}{2}\log(\eta_{\min}^2)+\log\left(\eta_{*}^2\right)+2\frac{\eta_{\max}^2}{\eta_{*}^2}+\frac{1}{2}\log\E{\eta^2(\{X_{\ell,k}\})}\nonumber\\
&\leq& 2\log(\pi)-\log(2)-\frac{1}{2}\log(\eta_{\min}^2)+\log\left(\eta_{*}^2\right)+2\frac{\eta_{\max}^2}{\eta_{*}^2}+\frac{1}{2}\log\left(\eta_{\max}^2\right)\nonumber\\
&=& 2\log(\pi)-\log(2)+\log\left(\eta_{*}^2\right)+\frac{1}{2}\log\left(\frac{\eta_{\max}^2}{\eta_{\min}^2}\right)
+2\frac{\eta_{\max}^2}{\eta_{*}^2}
\end{IEEEeqnarray}
where we have evaluated $\E{\left|\frac{Y_k}{\tilde{Y}_k}\right|}$ using \cite[p.69, Eq. (7.46)]{Simon2006PDI1212190} (see Appendix \ref{ExpectedValueRatio}) and we have upper-bounded $\E{\eta^2(\{X_{\ell,k}\})}\leq\eta^2_{\max}$.

Replacing (\ref{UB_1}), (\ref{UB_3}) and (\ref{UB_2}) in (\ref{MInf_2}), we finally obtain
\begin{IEEEeqnarray}{lCl}\label{MInf_4}
h(Y_1^n)-h(\tilde{Y}_1^n) 
&\leq& \sum_{k=1}^{n}\Bigg( \E{\log\bigg(\E{|\tilde{Y}_k|^2\Big|\{{X}_{\ell,k}\}}\bigg)}\Bigg.-\log(\pi e)-\E{\log\bigg(\E{|\tilde{Y}_k|^2\Big|\{{X}_{\ell,k}\}}\bigg)}
\nonumber \\
&&\qquad\qquad\qquad\qquad\qquad\quad{}\Bigg.+ 2\log(\pi)-\log(2)+\log\left(\eta_{*}^2\right)+\frac{1}{2}\log\left(\frac{\eta_{\max}^2}{\eta_{\min}^2}\right)
+2\frac{\eta_{\max}^2}{\eta_{*}^2}\Bigg)\nonumber\\
&=& \sum_{k=1}^{n}\bigg(\log(\pi)-1-\log(2)+\log\left(\eta_{*}^2\right)+\frac{1}{2}\log\left(\frac{\eta_{\max}^2}{\eta_{\min}^2}\right)+2\frac{\eta_{\max}^2}{\eta_{*}^2}\bigg)\nonumber\\
&=& n\mathsf{K}
\end{IEEEeqnarray}
where
\begin{IEEEeqnarray}{lCl}\label{K}
\mathsf{K}\stackrel{\Delta}{=}\log(\pi)-1-\log(2)+\log\left(\eta_{*}^2\right)
+\frac{1}{2}\log\left(\frac{\eta_{\max}^2}{\eta_{\min}^2}\right)
+2\frac{\eta_{\max}^2}{\eta_{*}^2}.
\end{IEEEeqnarray}
Back to the mutual information (\ref{MInf}), we have
\begin{IEEEeqnarray}{rCl}
\frac{1}{n}I({X}_{1}^n;Y_1^n)\leq \mathsf{K}
\end{IEEEeqnarray}
hence the capacity is upper-bounded by
\begin{equation}\label{BoundedC}
C(\SNR)\leq\mathsf{K}.
\end{equation}
This proves Theorem~\ref{MainTheorem}.

\section{Conclusions}
We analyzed the channel capacity of a noncoherent wireless network. We demonstrated that the channel capacity is bounded in the SNR, provided that the nodes do not cooperate and the path gains $\{\alpha_{\ell}\}$ decay exponentially or more slowly.
This confirms the observation made by Lozano et al.\ \cite{lozano_fundamental_2013} that in the absence of perfect CSI the capacity of wireless networks is bounded in the SNR. Our analysis is more general than the analysis in \cite{lozano_fundamental_2013} in the sense that we allow that the distribution of the channel inputs may depend on the SNR, but it is less general than the analysis in  \cite{lozano_fundamental_2013} in the sense that we assumed that the variances of the fading coefficients decay at most exponentially and that the nodes do not cooperate.

The behavior of the channel capacity observed in this paper is similar to the one observed for frequency-selective fading channels by Koch and Lapidoth in \cite{koch_multipath_2010}. Specifically, Koch and Lapidoth showed that if the path gains $\{\alpha_{\ell}\}$ decay exponentially or more slowly, then the capacity is bounded in the SNR; conversely, if the path gains decay faster than exponentially, then the capacity is unbounded in the SNR. The question then arises as to whether the channel capacity of noncoherent wireless networks is also unbounded in the SNR if we assume that the variances of the fading coefficients decay faster than exponentially. This will be the subject of future studies. 

\appendices
\section{Upper bound on (\ref{Integral2})}\label{Ap_Integralbyparts}
We solve the  second integral in (\ref{Integral2}) divided by $2\eta^2(\{x_{\ell,k}\})$, namely,
\begin{equation}\label{Integral2a}
\int_{0}^{\infty}\frac{2r}{(r^2+\eta^2{(\{x_{\ell,k}\})})^2}\log\left(r^2+\eta_*^2\right)dr.
\end{equation}
In the following, we first apply the change of variable $t= r^2$ to obtain
\begin{eqnarray}\label{Intparts_f(t)}
\int_{0}^{\infty}\frac{\log\left(t+\eta_{*}^2\right)}{\left(t+\eta^2{(\{x_{\ell,k}\})}\right)^2}dt.
\end{eqnarray}
We then integrate by parts. To this end, we distinguish between the two cases $\eta_*^2\geq\eta^2(\{x_{\ell,k}\})$ and  $\eta_*^2<\eta^2(\{x_{\ell,k}\})$.

\subsection{$\eta_*^2\geq\eta^2(\{x_{\ell,k}\})$}
We use another change of variable, namely, 
$x= t+\eta^2{(\{x_{\ell,k}\})}$, to obtain
\begin{IEEEeqnarray}{lCl}\label{Intparts_f(x)_c1}
\int_{0}^{\infty}\frac{\log\left(t+\eta_{*}^2\right)}{\left(t+\eta^2{(\{x_{\ell,k}\})}\right)^2}dt=\int_{\eta^2{(\{x_{\ell,k}\})}}^{\infty}\frac{\log\left(x+\eta_*^2-\eta^2{(\{x_{\ell,k}\})}\right)}{x^2}dx. \IEEEeqnarraynumspace
\end{IEEEeqnarray}
Introducing the auxiliary variable $\gamma=\eta_{*}^2-\eta^2{(\{X_{\ell,k}\})}$, this can be written as
\begin{IEEEeqnarray}{lCl}\label{Intparts_f(x)_c11}
\int_{\eta^2{(\{x_{\ell,k}\})}}^{\infty}\frac{\log\left(x+\eta_*^2-\eta^2{(\{x_{\ell,k}\})}\right)}{x^2}dx=\int_{\eta^2{(\{x_{\ell,k}\})}}^{\infty}\frac{1}{x^2}\log\left(x+\gamma\right)dx
\end{IEEEeqnarray}
which we now integrate by parts:
\begin{IEEEeqnarray}{lCl}\label{Integral2a1}
\int_{\eta^2{(\{x_{\ell,k}\})}}^{\infty}\frac{1}{x^2}\log\left(x+\gamma\right)dx&=&-\frac{1}{x}\log(x+\gamma)\bigg|_{\eta^2{(\{x_{\ell,k}\})}}^{\infty}+\int_{\eta^2{(\{x_{\ell,k}\})}}^{\infty}\frac{1}{x(x+\gamma)}dx\nonumber\\
&=& \frac{\log\left(\eta^2{(\{x_{\ell,k}\})}+\gamma\right)}{\eta^2{(\{x_{\ell,k}\})}}+\frac{1}{\gamma}\log\left(1+\frac{\gamma}{\eta^2{(\{x_{\ell,k}\})}}\right)\nonumber\\
&\leq&  \frac{\log\left(\eta^2{(\{x_{\ell,k}\})}+\gamma\right)}{\eta^2{(\{x_{\ell,k}\})}}+\frac{1}{\eta^2{(\{x_{\ell,k}\})}}\nonumber \\
&=& \frac{1}{\eta^2({\{x_{\ell,k}\}})}\left(\log\left(\eta_*^2\right)+1\right)
\end{IEEEeqnarray}
where the second step follows from \cite[p. 70, Sec 2.118, Eq. (1)]{Gradshteyn2007}, the third step follows from $\log(1+b)\leq b$ being $b=\frac{\gamma}{\eta_{(\{X_{\ell,k}\})}^2}$ and the last step follows by replacing $\eta^2{\{x_{\ell,k}\}}+\gamma=\eta_*^2$. Multiplying (\ref{Integral2a1}) by $2\eta^2(\{x_{\ell,k}\})$ and averaging over all values of $\{x_{\ell,k}\}$ yields (\ref{eta*>eta}).

\subsection{$\eta_*^2<\eta^2{(\{x_{\ell,k}\})}$}
We use again a change of variable, namely, $x=t+\eta_*^2$, to express the integral as
\begin{IEEEeqnarray}{lCl}\label{Intparts_f(x)_c2}
\int_{0}^{\infty}\frac{\log\left(t+\eta_{*}^2\right)}{\left(t+\eta^2{(\{x_{\ell,k}\})}\right)^2}dt=\int_{\eta_*^2}^{\infty}\frac{\log\left(x\right)}{\left(x+\eta^2{(\{x_{\ell,k}\})}-\eta_*^2\right)^2}dx.
\end{IEEEeqnarray}
Introducing the auxiliary variable $\gamma=\eta^2{(\{x_{\ell,k}\})}-\eta_*^2$ and integrating by parts yields
\begin{IEEEeqnarray}{lCl}\label{Integral2a2}
\int_{\eta_{*}^2}^{\infty}\frac{\log(x)}{(x+\gamma)^2}dx&=& -\frac{\log(x)}{(x+\gamma)}\bigg|_{\eta_*^2}^{\infty}+\int_{\eta_*^2}^{\infty}\frac{1}{x(x+\gamma)}dx\nonumber\\
&=& \frac{\log\left(\eta_*^2\right)}{(\eta_*^2+\gamma)}+\frac{1}{\gamma}\log\left(1+\frac{\gamma}{\eta_*^2}\right)\nonumber\\
&\leq&  \frac{1}{(\eta_*^2+\gamma)}\log\left(\eta_*^2\right)+\frac{1}{\eta_*^2}\nonumber \\
&=& \frac{1}{\eta^2{(\{x_{\ell,k}\})}}\left(\log\left(\eta_*^2\right)+\frac{\eta^2{(\{x_{\ell,k}\})}}{\eta_*^2}\right)
\end{IEEEeqnarray}
where the second step follows from \cite[p. 70, Sec. 2.118, Eq. (1)]{Gradshteyn2007}, the third step follows from $\log(1+b)\leq b$ being $b=\frac{\gamma}{\eta_{*}^2}$, and the last step follows by replacing $\eta^2{\{x_{\ell,k}\}}-\gamma=\eta_*^2$. Multiplying (\ref{Integral2a2}) by $2\eta^2(\{x_{\ell,k}\})$ and averaging over all values of $\{x_{\ell,k}\}$ yields (\ref{eta*<eta}). 
 
\section{Expected value of the ratio $\left|\frac{Y_k}{\tilde{Y}_k}\right|$}\label{ExpectedValueRatio}
In this section we calculate the expected value of the ratio $R=\left|\frac{Y_k}{\tilde{Y}_k}\right|$. Using the probability density function in (\ref{f_R}), we have
\begin{IEEEeqnarray}{lCl}\label{E{r}}
\E{R}&=&\int_{0}^{\infty}2r^2\E{\frac{\eta^2(\{x_{\ell,k}\})}{\left(r^2+\eta^2(\{x_{\ell,k}\})\right)^2}}dr\nonumber\\
&=&\E{\int_{0}^{\infty}\frac{2r^2\eta^2(\{x_{\ell,k}\})}{\left(r^2+\eta^2(\{x_{\ell,k}\})\right)^2}dr}\nonumber\\
&=& \E{2\eta^2(\{x_{\ell,k}\})\int_{0}^{\infty}\frac{r^2}{\left(r^2+\eta^2(\{x_{\ell,k}\})\right)^2}dr}\nonumber\\
\end{IEEEeqnarray}
where the second step follows by applying Fubini's theorem. Computing the integral yields
\begin{IEEEeqnarray}{lCl}\label{intGamma}
\int_{0}^{\infty}\frac{r^2}{\left(r^2+\eta^2(\{x_{\ell,k}\})\right)^2}dr&=&\frac{1}{2\eta^4\left(\{x_{\ell,k}\}\right)}\left({\eta^3\left(\{x_{\ell,k}\}\right)}\right)\frac{\Gamma\left(\frac{3}{2}\right)\Gamma\left(\frac{1}{2}\right)}{\Gamma\left(2\right)}\nonumber\\
&=&\frac{1}{2\eta\left(\{x_{\ell,k}\}\right)}\frac{\pi}{2},
\end{IEEEeqnarray}

where the first step follows from \cite[p. 322, Sec 3.241, Eq. (4)]{Gradshteyn2007}. Combining (\ref{intGamma}) with (\ref{E{r}}), we obtain the expected value of $\left|\frac{Y_k}{\tilde{Y}_k}\right|$:
\begin{eqnarray}
\E{\left|\frac{Y_k}{\tilde{Y}_k}\right|}=\E{R}=\frac{\pi}{2}\E{\eta\left(\{x_{\ell,k}\}\right)}.
\end{eqnarray}

\bibliographystyle{IEEEtran}


\begin{thebibliography}{10}
\providecommand{\url}[1]{#1}
\csname url@samestyle\endcsname
\providecommand{\newblock}{\relax}
\providecommand{\bibinfo}[2]{#2}
\providecommand{\BIBentrySTDinterwordspacing}{\spaceskip=0pt\relax}
\providecommand{\BIBentryALTinterwordstretchfactor}{4}
\providecommand{\BIBentryALTinterwordspacing}{\spaceskip=\fontdimen2\font plus
\BIBentryALTinterwordstretchfactor\fontdimen3\font minus
  \fontdimen4\font\relax}
\providecommand{\BIBforeignlanguage}[2]{{%
\expandafter\ifx\csname l@#1\endcsname\relax
\typeout{** WARNING: IEEEtran.bst: No hyphenation pattern has been}%
\typeout{** loaded for the language `#1'. Using the pattern for}%
\typeout{** the default language instead.}%
\else
\language=\csname l@#1\endcsname
\fi
#2}}
\providecommand{\BIBdecl}{\relax}
\BIBdecl

\bibitem{simeone_local_2009}
O.~Simeone, O.~Somekh, H.~Poor, and S.~Shamai, ``Local base station cooperation
  via finite-capacity links for the uplink of linear cellular networks,''
  \emph{IEEE Transactions on Information Theory}, vol.~55, no.~1, pp. 190--204,
  Jan. 2009.

\bibitem{ozgur_hierarchical_2007}
A.~Ozgur, O.~Leveque, and D.~Tse, ``Hierarchical cooperation achieves optimal
  capacity scaling in ad hoc networks,'' \emph{IEEE Transactions on Information
  Theory}, vol.~53, no.~10, pp. 3549--3572, Oct. 2007.

\bibitem{cadambe_interference_2008}
V.~Cadambe and S.~Jafar, ``Interference alignment and degrees of freedom of the
  k-user interference channel,'' \emph{IEEE Transactions on Information
  Theory}, vol.~54, no.~8, pp. 3425--3441, Aug. 2008.

\bibitem{irmer_coordinated_2011}
R.~Irmer, H.~Droste, P.~Marsch, M.~Grieger, G.~Fettweis, S.~Brueck, H.-P.
  Mayer, L.~Thiele, and V.~Jungnickel, ``Coordinated multipoint: {Concepts},
  performance, and field trial results,'' \emph{IEEE Communications Magazine},
  vol.~49, no.~2, pp. 102--111, Feb. 2011.

\bibitem{lozano_fundamental_2013}
A.~Lozano, R.~Heath, and J.~Andrews, ``Fundamental limits of cooperation,''
  \emph{IEEE Transactions on Information Theory}, vol.~59, no.~9, pp.
  5213--5226, Sep. 2013.

\bibitem{blachman1957communication}
N.~M. Blachman, ``Communication as a game,'' in \emph{Proc. IRE WESCON Conf},
  1957, pp. 61--66.

\bibitem{dobrushin1959optimum}
R.~Dobrushin, ``Optimum information transmission through a channel with unknown
  parameters,'' \emph{Radio Eng. Electron}, vol.~4, no.~12, pp. 1--8, 1959.

\bibitem{shomorony_worst-case_2013}
I.~Shomorony and A.~Avestimehr, ``Worst-case additive noise in wireless
  networks,'' \emph{IEEE Transactions on Information Theory}, vol.~59, no.~6,
  pp. 3833--3847, Jun. 2013.

\bibitem{shamai1992worst}
S.~Shamai and S.~Verdu, ``Worst-case power-constrained noise for binary-input
  channels,'' \emph{IEEE Transactions on Information Theory}, vol.~38, no.~5,
  pp. 1494--1511, 1992.

\bibitem{lapidoth2003capacity}
A.~Lapidoth and S.~M. Moser, ``Capacity bounds via duality with applications to
  multiple-antenna systems on flat-fading channels,'' \emph{{IEEE} Transactions
  on Information Theory}, vol.~49, no.~10, pp. 2426--2467, 2003.

\bibitem{Cover}
T.~M. Cover and J.~A. Thomas, \emph{Elements of Information Theory (Wiley
  Series in Telecommunications and Signal Processing)}.\hskip 1em plus 0.5em
  minus 0.4em\relax Wiley-Interscience, 2006.

\bibitem{koch_multipath_2010}
T.~Koch and A.~Lapidoth, ``On multipath fading channels at high {SNR},''
  \emph{{IEEE} Transactions on Information Theory}, vol.~56, no.~12, pp.
  5945--5957, Dec. 2010.

\bibitem{Simon2006PDI1212190}
M.~K. Simon, \emph{Probability Distributions Involving Gaussian Random
  Variables: A Handbook for Engineers, Scientists and Mathematicians}.\hskip
  1em plus 0.5em minus 0.4em\relax Secaucus, NJ, USA: Springer-Verlag New York,
  Inc., 2006.

\bibitem{ash2000probability}
R.~B. Ash and C.~Doleans-Dade, \emph{Probability and {Measure} {Theory}}.\hskip
  1em plus 0.5em minus 0.4em\relax Academic Press, 2000.

\bibitem{Gradshteyn2007}
I.~S. Gradshteyn and I.~M. Ryzhik, \emph{Table of {Integrals}, {Series}, and
  {Products}}, 7th~ed.\hskip 1em plus 0.5em minus 0.4em\relax Elsevier/Academic
  Press, Amsterdam, 2007.

\end{thebibliography}

\end{document}